\documentclass[sigconf]{acmart}

\setcopyright{rightsretained}

\AtBeginDocument{%
  \providecommand\BibTeX{{%
    \normalfont B\kern-0.5em{\scshape i\kern-0.25em b}\kern-0.8em\TeX}}}

\copyrightyear{2021}
\acmYear{2021}
\setcopyright{rightsretained}
\acmConference[FAccT '21]{Conference on Fairness, Accountability, and Transparency}{March 3--10, 2021}{Virtual Event, Canada}
\acmBooktitle{Conference on Fairness, Accountability, and Transparency (FAccT '21), March 3--10, 2021, Virtual Event, Canada}\acmDOI{10.1145/3442188.3445933}
\acmISBN{978-1-4503-8309-7/21/03}

\usepackage{subfigure}
\usepackage[utf8]{inputenc} % allow utf-8 input
\usepackage[T1]{fontenc}    % use 8-bit T1 fonts
\usepackage{hyperref}       % hyperlinks
\usepackage{url}            % simple URL typesetting
\usepackage{booktabs}       % professional-quality tables
\usepackage{amsmath,amsfonts,amsthm}  
\usepackage{multicol}
\usepackage{cleveref}
\usepackage{nicefrac}       % compact symbols for 1/2, etc.
\usepackage{microtype}      % microtypography
\usepackage{graphicx}
\usepackage{arydshln}
\usepackage{color}
\usepackage{dsfont}
\usepackage{appendix}
\usepackage{mathrsfs}

\newcommand{\prob}{\mathbb{P}}
\newcommand{\V}{V}

\DeclareMathOperator{\pa}{Pa}

\newcommand{\Value}{\mathtt{Value}}
\newcommand{\Click}{\mathtt{Click}}
\newcommand{\Reply}{\mathtt{Reply}}
\newcommand{\Open}{\mathtt{Open}}
\newcommand{\Fav}{\mathtt{Fav}}
\newcommand{\Favorite}{\mathtt{Favorite}}

\newcommand{\RT}{\mathtt{RT}}
\newcommand{\SLO}{\mathtt{SLO}}
\newcommand{\LC}{\mathtt{Link Click}}
\newcommand{\VW}{\mathtt{Vid Watch}}
\newcommand{\Linger}[1]{\mathtt{Linger >#1s}}
\newcommand{\Quote}{\mathtt{Quote}}
\newcommand{\UAM}{\mathtt{UAM}}
\newcommand{\Ntab}{\mathtt{NTab} \mathtt{View}}
\newcommand{\OptOut}{\mathtt{Opt} \mathtt{Out}}
\newcommand{\Bhvrs}{\mathtt{Behaviors}}
\newcommand{\Bhvr}{\mathtt{Behavior}}

\newcommand{\Msr}{A}
\DeclareMathOperator{\MB}{MB}
\DeclareMathOperator{\Pa}{Pa}

\newtheorem{theorem}{Theorem}
\newtheorem{corollary}{Corollary
}
\newtheorem{assumption}{Assumption
}

\usepackage{natbib}

\makeatletter
\def\thm@space@setup{%
  \thm@preskip=\parskip \thm@postskip=0pt
}
\makeatother

\begin{document}

\title{From Optimizing Engagement to Measuring Value}

\author{Smitha Milli}
\authornote{Work done while the author was an intern at Twitter.}
\affiliation{%
  \institution{UC Berkeley}
}
\email{smilli@berkeley.edu}

\author{Luca Belli}
\affiliation{%
  \institution{Twitter}
}
\email{lbelli@twitter.com}
\author{Moritz Hardt}
\authornote{†MH is a paid consultant at Twitter. Work performed while consulting for Twitter.}
\affiliation{%
  \institution{UC Berkeley}
}
\email{hardt@berkeley.edu}

\renewcommand{\shortauthors}{Smitha Milli, Luca Belli, and Moritz Hardt}

\begin{abstract}
Most recommendation engines today are based on predicting user engagement, e.g. predicting whether a user will click on an item or not. However, there is potentially a large gap between engagement signals and a desired notion of \emph{value} that is worth optimizing for. We use the framework of measurement theory to (a) confront the designer with a normative question about what the designer values, (b) provide a general latent variable model approach that can be used to operationalize the target construct and directly optimize for it, and (c) guide the designer in evaluating and revising their operationalization. We implement our approach on the Twitter platform on millions of users. In line with established approaches to assessing the \emph{validity} of measurements, we perform a qualitative evaluation of how well our model captures a desired notion of ``value''.
\end{abstract}

\maketitle

\defcitealias{standards2014}{AERA, APA, NCME, \citeyear{standards2014}}
\defcitealias{messick1987validity}{\citeauthor{messick1987validity}, \citeyear{messick1987validity}}
\defcitealias{reeves2016contemporary}{\citeauthor{reeves2016contemporary}, \citeyear{reeves2016contemporary}}

\section{Introduction} 

Most recommendation engines today are based on predicting user engagement, e.g. predicting whether a user will click an item or not. However, there is potentially a large gap between engagement signals and a desired notion of \emph{value} that is worth optimizing for \citep{ekstrand2016behaviorism}. Just because a user engages with an item doesn't mean they value it. A user might reply to an item because they are angry about it, or click an item in order to gain more information about it \citep{wen2019leveraging}, or watch addictive videos out of temptation. 

It is clear that engagements provide some signal for ``value'', but are not equivalent to it. Further, different types of engagement may provide differing levels of evidence for value. For example, if a user explicitly likes an item, we are more likely to believe that they value it, compared to if they had merely clicked on it. Ideally, we want the objective for our recommender system to take engagement signals into account, but only insofar as they relate to a desired notion of ``value''. However, directly specifying such an objective is a non-trivial problem. Exactly how much should we rely on likes versus clicks versus shares and so on? How do we evaluate whether our designed objective captures our intended notion of ``value''?

\subsection{Our contributions}

We make three primary contributions.

1. We propose measurement theory as a principled approach to aggregating engagement signals into an objective function that captures a desired notion of ``value''.  The resulting objective function can be optimized from data, serving as a plug-in replacement for the ad-hoc objectives typically used in engagement optimization frameworks.

2. Our approach is based on the creation of a latent variable model that relates value to various observed engagement signals. We devise a new identification strategy for the latent variable model tailored to the intended use case of online recommendation systems. Our identification strategy needs only a single robust engagement signal for which we know the conditional probability of value given the signal.

3. We implemented our approach on the Twitter platform on millions of users. In line with an established validity framework for measurement theory, we conduct a qualitative analysis of how well our model captures ``value''.

\subsection{Measurement theory and latent variable models}
The framework of \emph{measurement theory}~\citep{hand2004measurement,jackman2009} is widely used in the social sciences as a guide to measuring \emph{unobservable theoretical constructs} like ``quality of life'',  ``political ideology'', or ``socio-economic status''. Under the measurement approach, theoretical constructs are operationalized as latent variables, which are related to observable data through a latent variable model (LVM). 

Similarly, we treat the ``value'' of a recommendation as a theoretical construct, which we operationalize as a (binary) latent variable $V$. We represent the LVM as a a \emph{Bayesian network} ~\citep{pearl2009causality} that contains $V$ as well as each of the possible types of user engagements (clicks, shares, etc). The structure of the Bayesian network allows us to specify conditional independences between variables, enabling us to capture dependencies like e.g. needing to click an item before replying to it.

Under the measurement approach, the ideal objective becomes clear: $\prob(V=1 \mid \Bhvrs)$ - the probability the user values the item given their engagements with it. Such an objective uses all engagement signals, but only insofar provide evidence of Value $V$. If we can identify $\prob(V=1 \mid \Bhvrs)$, then it can be used as a drop-in replacement for any objective that scores items based on engagement signals.

%\Mnote{What may be confusing here is what this probability is taken over and why we can optimize that, i.e., why this is an ``objective function'' prima facie}

%\Mnote{It seems what's missing is a word about the latent variable model itself and the difficulty of coming up with a reasonable one.} 
Our key insight is that
we can identify $\prob(V \mid \Bhvrs)$ --- the probability of Value given \emph{all} behaviors --- through the use of a single \emph{anchor variable} $A$  for which we know $\prob(V = 1 \mid A = 1)$. The anchor variable, together with the structure of the Bayesian network, is what  gives ``value'' its meaning. Through the choice of the anchor variable and the structure of the Bayesian network, the designer has the flexibility to give ``value'' subtly different meanings.

Recommendation engines have natural candidates for anchor variables: strong, explicit feedback from the user. For example, strong negative feedback could include downvoting or reporting a content item, or blocking another user. Strong positive feedback could be explicitly liking or upvoting an item. For negative feedback, we  make the assumption that $\prob(V = 1 \mid A = 1) = \epsilon$ for $\epsilon \approx 0$, while for positive feedback we make the assumption that  $\prob(V = 1 \mid A = 1) = 1 - \epsilon$.

\subsection{A case study on the Twitter platform}

We implemented our approach on the Twitter platform on millions of users. On Twitter, there are numerous user behaviors: clicks, favorites, retweets, replies, and many more. It would be difficult to directly specify an objective that properly trades-off all these behaviors. Instead, we identify a natural anchor variable. On Twitter, users can give explicit feedback on tweets by clicking ``See less often'' (SLO) on them. We use SLO as our anchor and assume that the user does not value tweets they click ``See less often'' on. After specifying the anchor variable and the Bayesian network, we are able to learn $\prob(V \mid \Bhvrs)$ from data.

The model automatically learns a natural ordering of which behaviors should provide stronger evidence for Value $V$, e.g. $\prob(V = 1 \mid \mathtt{Retweet} = 1) > \prob(V = 1 \mid \mathtt{Reply} = 1) > \prob(V = 1 \mid \Click = 1)$. Furthermore, it learns complex inferences about the evidence provided by \emph{combinations} of behavior. Such inferences would not be possible under the standard approach, which uses a linear combination of behaviors as the objective.

Unlike other work on recommender systems, we do not evaluate through engagement metrics. If we believe that engagement is not the same as the construct ``value'', then we cannot evaluate our approach merely by reporting engagement numbers. Instead, we must take a more holisitc approach. We discuss established approaches to assessing the \emph{validity} \citep{standards2014,messick1987validity,reeves2016contemporary} of a measurement, and explain how they translate to the recommender system setting by using Twitter as an example.
\section{Related work}
In the social sciences, especially in psychology, education, and political science, measurement theory ~\citep{hand2004measurement} has long been used to operationalize constructs like ``personality'', ``intelligence'', ``political ideology'', etc. Often the operationalization of such constructs is heavily contested, and many types of evidence for validity and reliability are used to evaluate the match between a construct and its operationalization ~\citep{messick1987validity,standards2014}.

Recently, \citet{jacobs2019measurement} introduced the language of measurement in the context of computer science. They argue that many harms effected by computational systems are the direct result of a mis-match between a theoretical construct and its operationalization. In the context of recommender systems, many have argued that the engagement metrics used in practice are a poor operationalization of ``value'' \citep{ekstrand2016behaviorism}. 

We use measurement theory as a principled way to disentangle latent value from observed engagement. We provide a general latent variable model approach in which an \emph{anchor variable} provides the key link between the latent variable and the observed behaviors. The term anchor variable has been used been used in various ways in prior work on factor models \citep{arora2012learning,arora2013practical,halpern2016electronic,halpern2016clinical}; our usage is most similar to \citep{halpern2016clinical}. Our use of the anchor variable is also similar to the use of a \emph{proxy variable} to identify causal effects under unobserved confounding \citep{pearl2010measurement,kuroki2014measurement}. 

\section{Identification of the LVM with anchor $A$} \label{sec:identification}
We now describe our general approach to operationalizing a target construct through a latent variable model (LVM) with an \emph{anchor variable}. We operationalize the construct for value through a LVM in which the construct is represented through an unobserved, binary latent variable $V$ that the other binary, observed behaviors provide evidence for. We assume there is one observed behavior, an \emph{anchor variable} $A$, which we know $\prob(V = 1 \mid A = 1)$ for. We represent all other observed behaviors in the binary random vector $\mathbf{B} = (B_1, \dots, B_n)$. We refer to $A$ as an anchor variable because it will provide the crucial link to identifying $\prob(V \mid A, \mathbf{B})$. In other words, it will \emph{anchor} the other observed behaviors $\mathbf{B}$ to Value $V$.

We represent the LVM as a Bayesian network. A Bayesian network is a directed acyclic graph (DAG) that graphically encodes a factorization of the joint distribution of the variables in the network. In particular, the DAG encodes all conditional independences among the nodes through the $d$-separation rule~\citep{pearl2009causality}. This is important because in most real-world settings, the observed behaviors have complex dependencies among each other (e.g. one may need to click on an item before replying to it). Through our choice of the DAG we can model both the dependencies among the observed behaviors as well as the dependence of the unobserved variable $V$ on the observed behaviors.

Our goal is to determine $\prob(V \mid A, \mathbf{B})$ so that it can later be used downstream as a target for  optimization. We now discuss sufficient conditions for identifying the conditional distribution $\prob(V \mid A, \mathbf{B})$. There are three assumptions on the anchor variable $A$ that we will consider in turn.

\emph{Notation.} We use $\Pa(X)$ to denote the parents of a node $X$ and use $\Pa_{-V}(X) = \Pa(X)\setminus V$ to denote all parents of $X$ except for $V$.

\begin{assumption}[Value-sensitive] \label{as:value-sensitive}
For every realization $b$ of the random vector $\mathbf{B}$, we have that $\prob(A=1 \mid \mathbf{B} = b, \V = 1) \neq \prob(A=1 \mid \mathbf{B} = b, \V = 0)$.
\end{assumption}

Assumption \ref{as:value-sensitive} simply means that the anchor $A$ carries signal about Value $V$, regardless of what the other variables $\mathbf{B}$ are.\footnote{When combined with Assumption \ref{as:no-children}, Assumption \ref{as:value-sensitive} simplifies to the condition $\prob(A=1 \mid \pa_{-V}(A)=z, \V = 1) \neq \prob(A=1 \mid \pa_{-V}(A)=z, \V = 0)$ for every realization $z$ of $\pa_{-V}(A)$, the parents of $A$ excluding $V$.}
\begin{assumption}[No children] \label{as:no-children}
The anchor variable $A$ has no children. 
\end{assumption}

Since the anchor $A$ is chosen to be a strong type of explicit feedback, it is usually the last type of behavior the user engages in on a content item (e.g. a ``report'' button that removes the content from the user's timeline), and thus, it typically makes sense to model $A$ as having no children.

\begin{assumption}[One-sided conditional independence] \label{as:one-sided-ci}
Let $\pa_{-V}(A)$ be all parents of $A$ excluding $V$. Value $V$ is independent from $\pa_{-V}(A)$ given that $A=1$:
\begin{equation*}
    \prob(V = 1 \mid A = 1, \pa_{-V}(A)) = \prob(V = 1 \mid A = 1) \,.
\end{equation*} 
\end{assumption}

Assumption \ref{as:one-sided-ci} means that when the user has opted to give feedback ($A=1$), the level of information that feedback contains about Value $V$ does not depend on the other parents of $A$. The assumption rests on the fact that $A$ is a strong type of feedback that the user only provides when they are confident of their assessment. 

\subsection{Conditions for identification}
The next theorem establishes that under A\ref{as:value-sensitive}, the distribution of observable behaviors $\prob(A, \mathbf{B})$ and the conditional distribution $\prob(A \mid V, \mathbf{B})$ are sufficient for identifying the conditional distribution, $\prob(V \mid  A, \mathbf{B})$. The proof uses a \emph{matrix adjustment method} (\citeauthor{rothman2008modern}, \citeyear{rothman2008modern}; pg. 360) and is very similar to that in \citet{pearl2010measurement,kuroki2014measurement}.

\begin{theorem} \label{thm:obs-data-suff}
Let $\V$ and $A$ be binary random variables and let $\mathbf{B} = (B_1, \dots, B_n)$ be a binary random vector. If A\ref{as:value-sensitive} holds, then the distributions $\prob(A, \mathbf{B})$ and $\prob(A \mid V, \mathbf{B})$ uniquely identify the conditional distribution $\prob(\V \mid  A, \mathbf{B})$.
\end{theorem}
\begin{proof}
    Since the conditional distribution $\prob(\V \mid  A, \mathbf{B})$ is equal to $ \frac{\prob(\mathbf{B}, V) \cdot \prob(A \mid  \mathbf{B}, V)}{\prob(A, \mathbf{B})}$, we can reduce the problem to determining the distribution $\prob(\mathbf{B}, V)$. We can relate $\prob(\mathbf{B}, V)$ to the given distributions, $\prob(A, \mathbf{B})$ and $\prob(A \mid \mathbf{B}, V)$, via the law of total probability:
    \begin{align} \label{eq:non-matrix}
        \prob(A, \mathbf{B}) = \sum_{v\in \{0,1\}} \prob(\mathbf{B}, V=v) \prob(A \mid \mathbf{B}, V=v) \,.
    \end{align}
    For every realization $b$ of the random vector $\mathbf{B}$, we can write Equation \ref{eq:non-matrix} as $z^b = \mathbf{P}^{b} \mu^{b}$ where the matrix $\mathbf{P}^{b} \in [0, 1]^{2 \times 2}$ and the vectors $\mu^{b}, z^{b} \in [0, 1]^2$ are defined as
    \begin{align*}
        & \mathbf{P}^{b}_{i,j} = \prob(A = i \mid \mathbf{B} = b, V = j) ~\text{ for } ~i, j \in \{0, 1\}\,, \\
       & \mu^{b} = [\prob(\mathbf{B} = b, V = 0), \prob(\mathbf{B} = b, V = 1)]^T\,, \\
       & z^{b} = [\prob(\mathbf{B} = b, A = 0), \prob(\mathbf{B} = b, A = 1)]^T \,.
    \end{align*}
    Determining the distribution $\prob(\mathbf{B}, V)$ is equivalent to determining $\mu^{b}$ for all $b$. By Assumption \ref{as:value-sensitive}, for all $b$ we have $\prob(A = 1 \mid B = b, V = 1) \neq \prob(A = 1 \mid B = b, V = 0)$, which implies that the determinant of the matrix $\mathbf{P}^{b}$ is non-zero. Therefore, for all $b$, the vector $\mu^{b}$ is equal to $\mu^{b} = (\mathbf{P}^{b})^{-1} z^{b}$. Thus, $\prob(\mathbf{B}, V)$, and therefore the conditional distribution $\prob(\V \mid  A, \mathbf{B})$, is identified by the given distributions.
\end{proof}

If we add Assumption \ref{as:no-children}, i.e. the anchor $A$ has no children, then the distributions $\prob(A, \mathbf{B})$ and $\prob(A \mid \Pa(A))$ are sufficient to identify $\prob(V \mid A, \mathbf{B})$.

\begin{corollary} \label{cor:parents}
If the joint distribution $\prob(V, A, \mathbf{B})$ is Markov\footnote{A distribution $\prob(X_1, \dots, X_n)$ is said to be Markov with respect to a DAG $G$ if it factorizes according to $G$, i.e. $\prob(X_1, \dots, X_n) = \prod_{i \in [n]} \prob(X_i \mid \Pa(X_i))$.} with respect to a DAG $G$  in which A\ref{as:value-sensitive} and A\ref{as:no-children} hold, then the distributions  $\prob(A, \mathbf{B})$ and $\prob(A \mid \Pa(A))$ uniquely identify the conditional distribution $\prob(V \mid A, \mathbf{B})$.
\end{corollary}
\begin{proof}
In a Bayesian network, the \emph{Markov blanket} for a variable $X$ is the set of variables $\MB(X) \subseteq \mathcal{Z}$ that shield $X$ from all other variables $\mathcal{Z}$ in the DAG, i.e. $\prob(X \mid \mathcal{Z}) = \prob(X \mid \MB(X))$ \citep{pearl2009causality}. The Markov blanket for a variable $X$ consists of its parents, children, and parents of its children. Since the anchor $A$ has no children, $\prob(A \mid V, \mathbf{B}) = \prob(A \mid \MB(A)) = \prob(A \mid \Pa(A))$. Thus, by Theorem \ref{thm:obs-data-suff}, $\prob(A \mid \Pa(A))$, and $\prob(A, \mathbf{B})$ identify the conditional distribution $\prob(V \mid A, \mathbf{B})$
\end{proof}

Finally, when we add Assumption \ref{as:one-sided-ci}, one-sided conditional independence, then the distributions $\prob(V)$, $\prob(A, \mathbf{B})$, $\prob(V = 1 \mid A = 1)$, and $\prob(\Pa_{-V}(A) \mid V)$ are sufficient. The proof follows from Corollary \ref{cor:parents} because, under Assumption \ref{as:one-sided-ci}, the distributions $\prob(V = 1 \mid A = 1)$, $\prob(\Pa_{-V}(A) \mid V)$, and $\prob(V)$ identify $\prob(A \mid \Pa(A))$.

\begin{corollary} \label{cor:cond-ind}
If the joint distribution $\prob(V, A, \mathbf{B})$ is Markov with respect to a DAG $G$ in which A\ref{as:value-sensitive}-\ref{as:one-sided-ci} hold, then $\prob(V)$, $\prob(A, \mathbf{B})$, $\prob(V = 1 \mid A = 1)$, and $\prob(\Pa_{-V}(A) \mid V)$ uniquely identify the conditional distribution $\prob(V \mid A, \mathbf{B})$.
\end{corollary}
\begin{proof}
We will show that, under Assumption \ref{as:one-sided-ci}, the distributions $\prob(V = 1 \mid A = 1)$, $\prob(\Pa_{-V}(A) \mid V)$, and $\prob(V)$ identify $\prob(A \mid \Pa(A))$. The proof then follows from Corollary \ref{cor:parents}.

We show that we can identify $\prob(A \mid \Pa(A))$ by solving a set of linear equations. For short-hand let $p_{w,a,v} = \prob(\Pa_{-V}(A)=w, A=a, V=v)$. For any realization $w$, by marginalizing over $A$ and $V$, we can derive the following four equations for the four unknown probabilities $p_{w, 0, 0},\, p_{w, 0, 1},\, p_{w, 1, 0},\, p_{w, 1, 1}$:
\begin{align}
    & \prob(\Pa_{-V}(A)=w, A=0) =  p_{w,0,0} + p_{w,0,1} \label{eq:m=0} \\
    & \prob(\Pa_{-V}(A)=w, A=1) = p_{w,1,0} + p_{w,1,1} \label{eq:m=1} \\
    & \prob(\Pa_{-V}(A)=w, V=0) =  p_{w,0,0} + p_{w,1,0} \label{eq:v=0} \\
    & \prob(\Pa_{-V}(A)=w, V=1) =  p_{w,0,1} + p_{w,1,1} \label{eq:v=1}
\end{align}
Note that the LHS of Equations \ref{eq:m=0} and \ref{eq:m=1} are given by $\prob(A, \mathbf{B})$ and the LHS of Equations \ref{eq:v=0} and \ref{eq:v=1} are given by the prior $\prob(V)$ and $\prob(\pa_{-V}(A) \mid V)$.

From Assumption 3, one-sided conditional independence, we know that $\prob(V = 1 \mid A = 1, \Pa_{-V}(A)) = \prob(V = 1 \mid A = 1)$. Under one-sided conditional independence, the probability $p_{w,1,1}$ is determined by the given distributions:
\begin{align}
    p_{w,1,1} =~ & \prob(A = 1) \cdot \prob(\Pa_{-V}(A)=w\mid A=1) \nonumber \\ & \cdot \prob(V = 1 \mid A = 1, \Pa_{-V}(A)=w) \nonumber\\
    =~ & \prob(A=1)\cdot \prob(\Pa_{-V}(A)=w\mid A=1) \nonumber \\
    & \cdot \prob(V = 1 \mid A = 1)\,. \label{eq:pw11}
\end{align}
Since $p_{w,1,1}$ is determined by the given distributions, so are $p_{w, 0, 0}$, $p_{w, 1, 0}$, and $p_{w, 0, 1}$, which can be solved for through Equations \ref{eq:m=0}-\ref{eq:v=1}. Since this holds for any realization $w$, the distribution $\prob(A, V, \Pa_{-V}(A)) = \prob(A, \Pa(A))$ is determined, which by Collorary \ref{cor:parents} means that the conditional distribution $\prob(V \mid A, \mathbf{B})$ is determined.
\end{proof}

\subsection{Specifying the distributions for identification} \label{sec:spec-dists}
Corollary \ref{cor:cond-ind} establishes that, under Assumptions \ref{as:value-sensitive}-\ref{as:one-sided-ci}, the distributions $\prob(V)$, $\prob(A, \mathbf{B})$, $\prob(\Pa_{-V}(A) \mid V)$, and $\prob(A = 1 \mid V = 1)$ are sufficient to determine the conditional distribution $\prob(V \mid A, \mathbf{B})$ for the LVM. Where do we get these distributions?

\begin{enumerate}
    \item The distribution of observable nodes $\prob(A, \mathbf{B})$ is estimated by the empirical distribution of observed data.
    \item The distribution $\prob(V)$ over the latent variable for value $V$ is a prior distribution that is specified by the modeler. Recall that our goal with the LVM is to use $\prob(V = 1 \mid 
\mathbf{B}, A)$ as an objective to optimize. Since the prior $\prob(V)$ only has a scaling effect on $\prob(V = 1 \mid 
\mathbf{B}, A)$, it does not matter greatly. We set $\prob(V)$ to be uniform, i.e. $\prob(V=1)=0.5$. 

\item The conditional probability $\prob(V = 1 \mid A = 1)$ is specified by our assumption on the the anchor variable. The probability $\prob(V = 1 \mid A = 1)$ is set to $\epsilon$ where $\epsilon \approx 0$ if $A$ is explicit negative feedback or to $1-\epsilon$ if $A$ is explicit positive feedback. 

\item That leaves the distribution $\prob(\Pa_{-V}(A) \mid V)$. We estimate $\prob(\Pa_{-V}(A) \mid V)$ heuristically using two sources of historical data that vary in their distribution of Value $V$. Suppose we have access to a dataset of historical recommendations $\mathcal{D}_{R}$ that were sent to users at random, as well as a dataset of historical recommendations that were algorithmically chosen, $\mathcal{D}_{C}$. Both kinds of datasets are commonly available on recommender systems due to the prevalence of A/B testing which typically tests new algorithmic changes against a randomized baseline. The randomized and algorithmic datasets will have different distributions of valuable content, $\prob_{R}(V)$ and $\prob_{C}(V)$, and different distributions of observed behavior, $\prob_{R}(A, \mathbf{B})$ and $\prob_{C}(A, \mathbf{B})$. However, we assume that $\prob(A, \mathbf{B} \mid V)$, the probability of the observed behavior given Value $V$, is the same between the two datasets.\footnote{If our DAG has Value $V$ as a root node and can be interpreted as a causal Bayesian network \citep{pearl2009causality}, then this is equivalent to assuming that the difference between the datasets corresponds to an intervention on $V$.}  The following equations then hold:
\begin{align}
    \prob_R(\Pa_{-V}(A)) & = ~\prob(\Pa_{-V}(A) \mid V = 1)\prob_R(V=1) \nonumber \label{eq:random-data} \\
    & + \prob(\Pa_{-V}(A) \mid V = 0)\prob_R(V=0) \,, \\
    \prob_C(\Pa_{-V}(A)) & =  ~\prob(\Pa_{-V}(A) \mid V = 1)\prob_C(V=1) \nonumber \label{eq:algo-data} \\ 
    & + \prob(\Pa_{-V}(A) \mid V = 0)\prob_C(V=0)\,.
\end{align}
 We specify $\prob_R(V)$ and $\prob_C(V)$ in an application-dependent way, but, generally, we assume the randomized dataset is lower value than the algorithmic one: $\prob_R(V) < \prob_C(V)$. Once we specify $\prob_R(V)$ and $\prob_C(V)$ and estimate $\prob_{R}(A, \mathbf{B})$ and $\prob_{C}(A, \mathbf{B})$ empirically, then we can solve Equations \ref{eq:random-data} and \ref{eq:algo-data} to estimate $\prob(\Pa_{-V}(A) \mid V = 1)$. This is a heuristic approach that is appropriate for getting a rough estimate, but needs to be used with care. In practice, not all the differences between the randomized and algorithmic dataset can be explained by an intervention on Value $V$. For example, if the recommendation algorithm has historically been optimized for user clicks, then users in the algorithmic dataset may click on items more, but for reasons other than increased value.
\end{enumerate}

\subsection{Algorithm for identification}
We now give more details on how we calculate the joint distribution $\prob(V, A, \mathbf{B})$ given the distributions $\prob(V)$, $\prob(A, \mathbf{B})$, $\prob(V = 1 \mid A = 1)$ and $\prob(\pa_{-V}(A) \mid V)$. We use the structure of the Bayesian network to efficiently identify the joint distribution $\prob(V, A, \mathbf{B})$ by fitting each factor $\prob(X \mid \Pa(X))$ for every variable $X$. 

\begin{enumerate}
    \item The factor for $V$ is given by the prior $\prob(V)$.\footnote{Assuming that $V$ is a root node, which is the case in any network we are interested in.}
    \item The factor for $A$, i.e. $\prob(A \mid \Pa(A))$, can be identified from $\prob(V)$, $\prob(V = 1 \mid A = 1)$, and $\prob(\Pa_{-V}(A) \mid V)$ by solving a set of linear equations as in the proof of Corollary \ref{cor:cond-ind}.
    \item   The factor for any behavior that does not have $V$ as a parent is directly identified by the distribution of observable behaviors $\prob(A, \mathbf{B})$. 
    \item The factors for the remaining behaviors which have $V$ as a parent are fit through a \emph{matrix adjustment method} (\citeauthor{rothman2008modern}, \citeyear{rothman2008modern}; pg. 360). In particular, note that
    {\small
    \begin{align*}
        & \prob(X=1, \Pa_{-V}(X) = z, \Pa_{-V}(A)=w, A = a) = \\
        &  ( \sum_{v \in \{0, 1\} }\prob(A = a \mid \Pa_{-V}(A) = w, V=v) \\
        & \cdot \prob(X=1, \Pa_{-V}(X)=z, \Pa_{-V}(A) = w, V=v) )
    \end{align*}
    }%
    
    We can also write the above equation in matrix form. Let $z_1, \dots, z_m$ be all realizations of $\Pa_{-V}(X)$, and define the matrices $Q^{w} \in [0,1]^{2\times m}$, $R^{w}\in [0,1]^{2\times2}$, $S^{w} \in [0,1]^{2\times m}$ as\footnote{If $\Pa_{-V}(X) \cap \Pa_{-V}(A) \neq \emptyset$ and $\Pa_{-V}(X)=z_i$ and $\Pa_{-V}(A)=w$ conflict, then simply set $Q^{w}_{0,i} = Q^{w}_{1, i} = 0$.} 
    {\small \begin{align}
        & Q^{w}_{a,i} = \prob(X=1, \Pa_{-V}(X) = z_i, \\ 
        &
        \qquad \qquad \Pa_{-V}(A)=w, A = a),\nonumber \\
        & R^{w}_{av} = \prob(A = a \mid \Pa_{-V}(A) = w, V=v), \\
        & S^{w}_{v,i} = \prob(X=1, \Pa_{-V}(X)=z_i, \\
        & \qquad \qquad \Pa_{-V}(A) = w, V=v). \nonumber
    \end{align}
    }%
    
    Then, $Q^{w} = R^{w}S^{w}$ and $S^{w} = (R^{w})^{-1}Q^{w}$.\footnote{$R^{w}$ is invertible because of Assumption \ref{as:value-sensitive}.} Let $S$ be the marginalization over $w$: $\sum_{w} S^{w} = (R^{w})^{-1}Q^{w}$. Then $S_{v,i} = \prob(X=1, \Pa_{-V}(X)=z_i, V=v)$. Thus, the factor for $X$ is equal to $\prob(X \mid \Pa_{-V}(X)=z_i, V=v) = S_{v,i}/\prob(\Pa_{-V}(X)=z_i, V=v)$. We fit nodes with $V$ as a parent in topological order, so that we can always calculate the denominator from previously fit factors.
    
\end{enumerate}

\begin{figure*}[t]
    \centering
    \includegraphics[width=\columnwidth]{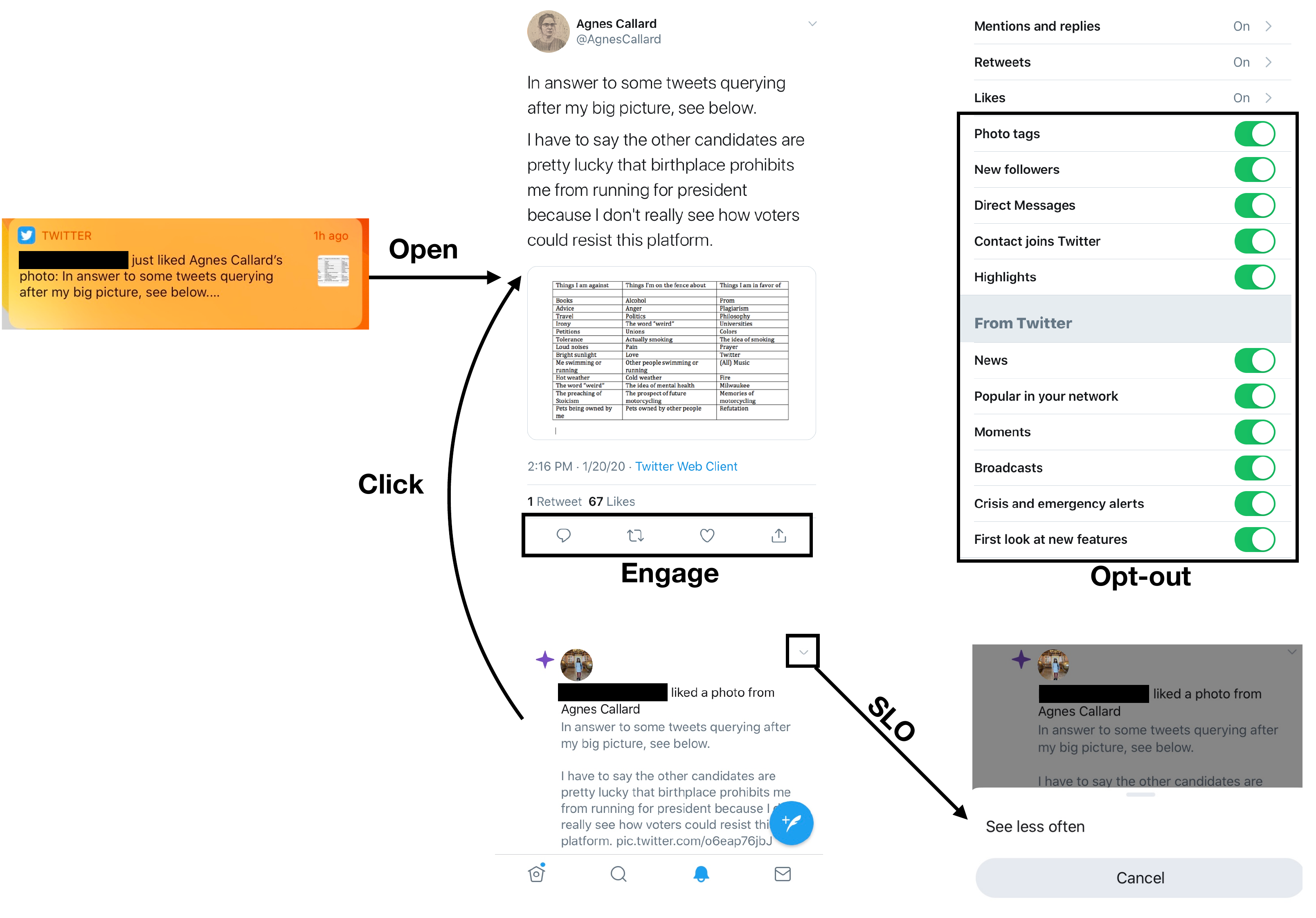}
    \caption{A workflow of how users can interact with ML-based notifications on Twitter. To view the tweet, the user can either ``open'' the notification from the home screen on their phone or ``click'' on it from the notifications tab within the app. If the user sees the tweet from their notifications tab, they can also click "See Less Often" on it. Once the user has opened or clicked on the notification, they can engage with the tweet in many ways, e.g. replying, retweeting, or favoriting. At any point, the user can opt-out of notifications all-together.}
    \label{fig:mr_workflow}
\end{figure*}

\begin{figure*}[t]
\centering
\includegraphics[width=0.7\textwidth]{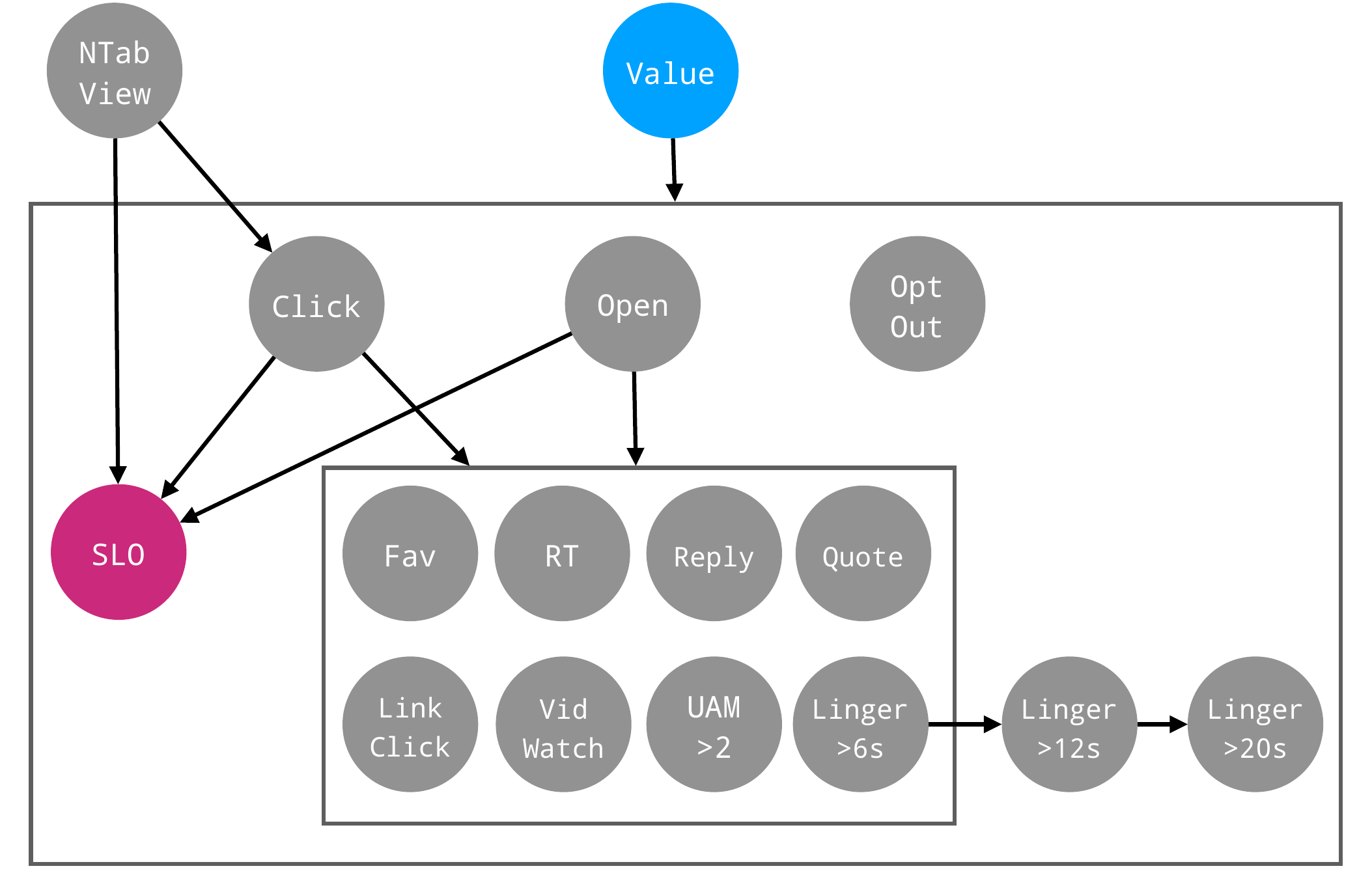}
\caption{Bayesian network for Twitter notifications. An arrow from a node $X$ to a box means that the node $X$ is a parent of all the nodes in the box, e.g. $\Click$ and $\Open$ are parents of $\Fav$, $\RT$, ..., $\Linger{6}$. The latent variable $\Value$ is a parent of everything except $\Ntab$. The measurement node $\SLO$  is highlighted in pink.} \label{fig:bnet}
\end{figure*}

\section{Application to Twitter} \label{sec:twitter-case-study}
We implemented our approach on the Twitter platform on millions of users. On Twitter, there are many kinds of user behaviors: clicks, replies, favorites, retweets, etc. The typical approach to recommendations would involve optimizing an objective that trades-off these behaviors, usually with linear weights. However, designing an objective is a non-trivial problem. How exactly should we weigh favorites compared to clicks or replies or retweets or any of the numerous other behaviors? It is difficult to assess whether the weights we chose match the notion of ``value'' we intended.

Furthermore, even supposing that we could manually specify the ``correct'' weights through laborious trial-and-error, the correct weights change over time. For example, after videos shared on Twitter began to auto-play, the signal of whether or not a user watched a video presumably became less relevant. The reality is that the objective is never static - how users interact with the platform is constantly changing, and the objective must change accordingly.

Our approach provides a principled solution to objective specification. We directly operationalize our intended construct ``value'' as a latent variable $V$. The meaning of Value $V$ is defined by the Bayesian network and the \emph{anchor variable} $A$, a behavior that we believe provides strong evidence for value or the lack of it. On Twitter, the user can provide strong, explicit feedback by clicking ``See less often'' (SLO) on a tweet. We use SLO as our anchor $A$ and assume that if a user clicks "See less often" on a tweet, they do not value it:  $\prob(V = 1 \mid \SLO = 1) = 0$.

Under this approach, there is no need to manually specify how all the behaviors should factor into the objective. Having operationalized Value, the ideal objective to use is clear: $\prob(V = 1 \mid \mathbf{B}, A)$ - the probability of Value $V$ given the observed behaviors. As discussed in Section \ref{sec:twitter-case-study}, we can directly estimate $\prob(V = 1 \mid \mathbf{B}, A)$ from data. Furthermore, presuming that the anchor and structure of Bayesian network remain stable, we can regularly re-estimate the model with new data at any point, allowing us to account for change in user behavior on the platform.

\textbf{The Bayesian network.}
We applied our approach to ML-driven notifications on Twitter. These notifications have various forms, e.g. "Users A, B, C just liked User Z's tweet", "User A just tweeted after a long time", or "Users A, B, C followed User Z". Figure \ref{fig:mr_workflow} shows an example notification and how a user can interact with it. The Bayesian network in Figure \ref{fig:bnet} succinctly encodes the dependencies between different types of interactions users can have with notifications.\footnote{The network can be interpreted as a \emph{causal} Bayesian network \citep{pearl2009causality}, although for our purposes, we do not strictly need the causal interpretation.}

Notifications are sent both to the user's home screen on their mobile phone, as well as to the notifications tab within the Twitter app. The user can start their interaction either by seeing the notification in their notification tab ($\Ntab$), and then clicking on it ($\Click$), or by seeing it as a the notification on their phone home screen and opening it from there directly ($\Open$). After clicking or opening the notification, the user can engage in many more interactions: they can favorite ($\Fav$), retweet ($\RT$), quote retweet ($\Quote$), or reply ($\Reply$) to the tweet; if the tweet has a link, they can click on it ($\LC$); if it has a video, they can watch it ($\VW$). In addition, other implicit signals are logged: whether the amount the user lingered on the tweet exceeds certain thresholds ($\Linger{6}$, $\Linger{12}$, $\Linger{20}$) and whether the number of user active minutes ($\UAM$) spent in the app after clicking/opening the notification exceeds a threshold.

Furthermore, when the user is in the notification tab, the user can provide explicit feedback on a particular notification by clicking "See Less Often" ($\SLO$) on it. Notably, unlike other types of behavior, the user does not need to actually click or open the notification before clicking SLO. However, we found empirically that users are more likely to click SLO after clicking or opening the notification, probably because they need to gain more information before making an assessment. Thus, in addition to $\Ntab$, we also model $\Click$ and $\Open$ as parents of $\SLO$.

Finally, at any time the user can opt-out of notifications to their phone home screen ($\OptOut$). When the user decides to opt-out, it is attributed to any ML-based notification saw within a day of choosing to opt-out. Since ML-based notifications are relatively rare on Twitter (users usually get less than one a day), there are usually at most one or two notifications attributed to an opt-out event.

We model the latent variable $V$ as being a parent of all behaviors except $\Ntab$ (whether or not the user saw the notification in their notifications tab or not). Since users may check their notifications tab for many other notifications, it is difficult to attribute $\Ntab$ to a particular notification, and so we consider it to be an exogenous, random event.

\textbf{Identifying the joint distribution} We fit our model on three days of data that contained of all user interactions with ML-based push notifications on Twitter. In Section \ref{sec:identification}, we proved that the target objective - the conditional distribution $\prob(V = 1 \mid \mathbf{B}, A)$ - is uniquely identified from $\prob(V = 1 \mid \Msr = 1)$, $\prob(V)$, $\prob(\mathbf{B}, A)$, and $\prob(\Pa_{-V}(A) \mid A)$ (see Corollary \ref{cor:cond-ind}). We set the four distributions as follows. We used $\SLO$ as our anchor variable $\Msr$ and assumed that $\prob(V = 1 \mid \Msr = 1) = 0$, i.e. a user never says ``See less often'' if they value the notification. The prior distribution of value $\prob(V)$ was set to be uniform. The distribution of observed behaviors $\prob(\mathbf{B}, A)$ was set to the empirical distribution. The distribution $\prob(\Pa_{-V}(A) \mid V)$ was estimated as described in Section \ref{sec:spec-dists} by using two sources of historical data, one in which notifications were sent at random and the other in which notifications were sent according to a recommendation algorithm.\footnote{We assume that the dataset of randomized notifications has a prior probability $\prob_{R}(V=1) = 0$ and the dataset of algorithmically chosen notifications has a prior probability $\prob_C(V=1) = 0.5$.}

\begin{table*}[t]
    \centering
    {\renewcommand{\arraystretch}{1.3}
    \begin{tabular}{|l|c|c|c|} \multicolumn{4}{c}{$\prob(V=1 \mid \mathtt{Behavior}=1)$} \\
        \bottomrule
          $\mathtt{Behavior}$ & Naive Bayes & $\Click, \Open \nrightarrow \SLO$ & Full Model \\
          \bottomrule
         $\OptOut$ & 0 & 0 & 0 \\ \hline
         $\Click$ & 0 & 0.316 & 0.652 \\ \hline
         $\Open$ & 0 & 0.442 & 0.685 \\ \hline
         $\UAM$ & 0 & 0.157 & 0.719 \\ \hline
         $\VW$ & 0 & 0.254 & 0.772\\ \hline
         $\Linger{6}$ & 0 & 0.264 & 0.802 \\ \hline
         $\LC$ & 0 & 0.320 & 0.836 \\ \hline
         $\Reply$ & 0.358 & 0.570 & 0.932 \\ \hline
         $\Linger{12}$ & 0 & 0.245 & 0.948\\ \hline
         $\Fav$ & 0.579 & 0.672 & 0.949 \\ \hline
         $\RT$ & 0.680 & 0.720 & 0.956\\ \hline
         $\Linger{20}$ & 0.019 & 0.296 & 0.991 \\
         \hline
         $\Quote$ & 1.0 & 1.0 & 1.0 \\
         \bottomrule
    \end{tabular}
    }
    \caption{The inferences made by LVMs with different DAGs. For each model and for each behavior, we list $\prob(V = 1 \mid \Bhvr = 1)$ -- how much evidence the model learns that a behavior provides for Value $V$ (when all other behaviors are marginalized over).}
    \label{tab:bnet_infs}
\end{table*}

\textbf{Evaluation of internal structure.} Assessing our measure of ``value'' for validity will necessarily be an on-going and multi-faceted process. We do not, as typical of papers on recommendation, report engagement metrics. The reason is that if we expect our measure of ``value'' to differ from engagement, we cannot evaluate it by simply reporting engagement metrics. The evaluation of a measurement necessitates a more holistic approach. In Section \ref{sec:validity}, we describe the five categories of evidence for validity described by the \emph{Standards for educational and psychological testing}, the handbook considered the gold standard on approaches to testing \citep{standards2014}.

Here, we focus on evaluating what is known as \emph{evidence based on internal structure}, i.e whether expected theoretical relationships between the variables in the model hold. To justify why the structure of our Bayesian network is necessary, we compare our full model from Figure \ref{fig:bnet} to two other models: a naive Bayes model and the full model but without arrows from $\Open$ and $\Click$ to $\SLO$. In Table \ref{tab:bnet_infs}, we show $\prob(V = 1 \mid \Bhvr = 1)$ for all behaviors and models. As noted by prior work \citep{pearl2009causality,halpern2016clinical}, matrix adjustment methods  can result in negative values when conditional independence assumptions are not satisfied. To address this, we clamp all inferences to the interval $[0, 1]$. We include the table of non-clamped inferences in the appendix (Table \ref{tab:bnet_infs_unclamped}).

The first, simple theoretical relationship we expect to hold is that compared to observing no user interaction, observing any user behavior besides opt-out should increase the probability that the user values the tweet, i.e. $\prob(V = 1 \mid \Bhvr = 1) < \prob(V = 1) = 0.5$ for all $\Bhvr \neq \OptOut$. Furthermore, we also expect some behaviors to provide stronger signals of value than others, e.g. that  $\prob(V = 1 \mid \Fav = 1) > \prob(V = 1 \mid \Click = 1)$. 

The first model is the naive Bayes model, which simply assumes that all behaviors are conditionally independent given Value $V$. It does extremely poorly - almost all inferences have negative values and are clamped to zero, indicating that the conditional independence assumptions are unrealistic.

The second model is the full model except without arrows from $\Click$ and $\Open$ to $\SLO$. It models all pre-requisite relationships between behaviors, i.e. if a behavior $X$ is required for another behavior $Y$, then there is an arrow from $X$ to $Y$. Compared to the naive Bayes model, the second model does not make mainly negative-valued inferences, indicating that its conditional independence assumptions are more realistic. However, relative to the prior, most behaviors actually reduce the probability of $\Value$, rather than increase it!

After investigation, we realized that although users were not technically \emph{required} to click or open the notification before clicking SLO, in practice, they were more likely to do so, probably because they needed to gain information before making an assessment. We found that explicitly modeling the connection, i.e. adding arrows from $\Click$ and $\Open$ to $\SLO$ was critical for making reasonable inferences. We believe this takeaway will apply across recommender systems. The user never has perfect information and may need to engage with an item before providing explicit feedback \citep{wen2019leveraging}. It is important to model the relationship between information-gaining behavior and explicit feedback in the Bayesian network.

Our full model satisfies the theoretical relationships we expect. All the behaviors that we expect to increase the probability of Value $V$ do indeed do so. Furthermore, the relative strength of different types of behavior seems reasonable as well, e.g. $\prob(V = 1 \mid \Fav = 1)$ and  $\prob(V= 1 \mid \RT = 1)$ are higher than $\prob(V = 1 \mid \VW = 1)$ and $\prob(V = 1 \mid \LC = 1)$.

The full model also makes more nuanced theoretical inferences. Recall that $\UAM$ is whether or not the user had high user active minutes after either clicking the notification from notifications tab or by opening the notification from their phone home screen. The model learns that $\UAM$ is a highly indicative signal after $\Open$, but not after $\Click$: $\prob(V = 1 \mid \Open = 1, \UAM = 1) = 0.906$ and $\prob(V = 1 \mid \Click = 1, \UAM = 1) = 0.641$. This makes sense because if the user clicks from notifications tab, it means they were already in the app, and it is difficult to attribute their high UAM to the notification in particular. On the other hand, if the user enters the app because of the notification, it is much more direct of an attribution.

It is clear that manually specifying the inferences our model makes would be very difficult. The advantage of our approach is that after specifying (a) the anchor variable and (b) the Bayesian network, we can automatically learn these inferences from data. Further, the model is able to learn complex inferences (e.g. that $\UAM$ is more reliable after $\Open$ than $\Click$) that would be impossible to specify under the typical linear weighting of behaviors.
\section{Assessing validity} \label{sec:validity}
Thus far, we have described our framework for designing a measure of ``value'', which can be used as a principled replacement for the ad-hoc objectives ordinarily used in engagement optimization. How do we evaluate such a measure?  Notably, we do not advocate evaluating the measure purely through engagement metrics. If we expect our measure of ``value'' to differ from engagement, then we cannot evaluate it by simply reporting engagement metrics. Instead, the assessment of any measure is necessarily an ongoing, multi-faceted, and interdisciplinary process. 

To complete the presentation of our framework, we now discuss approaches to assess the \emph{validity} \citep{messick1987validity,standards2014,reeves2016contemporary} of a measurement. In the most recent (\citeyear{standards2014}) edition of the \emph{Standards for educational and psychological testing}, the handbook considered the gold standard on approaches to testing, there are five categories of evidence for validity \citep{standards2014}. We visit each in turn, and describe how they translate to the recommender system setting, using Twitter as an example.

\textbf{Evidence based on content} refers to whether the content of a measurement is sufficient to fully capture the target construct. For example, we may question whether a measure of ``socio-economic status'' that includes income, but does not account for wealth, accurately captures the content of the construct \citep{jacobs2019measurement}. In the recommender engine setting, content-based evidence asks us to reflect on whether the behaviors available on the platform are sufficient to capture a worthy notion of the construct ``value''. For example, if the only behavior observed on the platform were clicks by the user, then we may be skeptical of any measurement of ``value'' derived from user behavior. What content-based evidence makes clear is that to measure any worthy notion of ``value'', it is essential to design platforms in which users are empowered with richer channels of feedback. Otherwise, no measurement derived from user behavior will accurately capture the construct.

\textbf{Evidence based on cognitive processes.} Measurements derived from human behavior are often based on implicit assumptions about the cognitive processes subjects engage in. Cognitive process evidence refers to evidence about such assumptions, often derived from explicit studies with subjects. For example, consider a reading comprehension test. We assume that high-scoring students succeed by using critical reading skills, rather than a superficial heuristic like picking the answers with the longest length. To gain evidence about whether this assumption holds, we might, for instance, ask students to take the test while verbalizing what they are thinking.

Similarly, in the recommender engine setting, we want to verify whether user behaviors occur for the reasons we think they do. On Twitter, one might think to use $\Favorite$ as an anchor for Value $V$, assuming that $\prob(V = 1 \mid \Favorite = 1) \approx 1$. However, users actually favorite items for reasons that may not reflect value -- like to bookmark a tweet or to stop a conversation. Cognitive process evidence highlights the importance of user research in assessing the validity of any measure of ``value''.

\textbf{Evidence based on internal structure} refers to whether the observations the measurement is derived from conform to expected, theoretical relationships. For example, for a test with questions which we expect to be of increasing difficulty, we would assess whether students actually perform worse on later questions, compared to earlier ones. In the recommender system context, we may have expectations on which types of user behaviors should provide stronger signal for value. In Section \ref{sec:twitter-case-study}, we evaluated internal structure by comparing $\prob(V = 1 \mid \Bhvr = 1)$ for all behaviors.

\textbf{Evidence based on relations with other variables} is concerned with the relationships between the measurement and other variables that are external to the measurement. The external variables could be variables which the measurement is expected to be similar to or predict, as well as variables which the measurement is expected to differ from. For example, a new measure of depression should correlate with other, existing measures of depression, but correlate less with measures of other disorders. In the recommender system context, we might look at whether our derived measurement of ``value'' is predictive of answers that users give in explicit surveys about content they value. We could also verify that our measure of ``value'' does not differ based on protected attributes, like the sex or race of the author of the content. 

\textbf{Evidence based on consequences.} Finally, the consequences of a measurement cannot be separated from its validity. Consider a test to measure student mathematical ability. The test is used to sort students into beginner or advanced classes with the hypothesis that all students will do better after sorted into their appropriate class. If it turns out that students sorted by the test do \emph{not} perform better, that may give us reason to reassess the original test. In the recommender system context, if we find that after using our measurement of value to optimize recommendations, more users complain or quit the platform, then we would have reason to revise our measurement.
\section{Summary}
We have presented a framework for designing an objective function that captures a desired notion of ``value''. In line with the principles of measurement theory, we treat ``value'' as a theoretical construct which must be operationalized. Our framework allows the designer to operationalize ``value'' in a principled manner by specifying only an \emph{anchor variable} and the structure of the Bayesian network. Through these two choices, the designer has the flexibility to give ``value'' subtly different meanings. 

We applied our approach on the Twitter platform on millions of users. We do not, as typical of papers on recommendation, report engagement metrics. The reason is that if we expect our measure of ``value'' to differ from engagement, we cannot evaluate it simply by reporting engagement metrics. Instead, we discussed established ways to assess the validity of a measurement and how they translate to the recommendation system setting. For the scope of this work, we focused on assessing \emph{evidence based on internal structure} and found that our measure of ``value'' satisfied many desired theoretical relationships.
\section*{Acknowledgements}
We thank Naz Erkan for giving us the opportunity and freedom to conduct this work through her bold leadership and savvy managerial support. We thank Prakhar Biyani for his extensive effort in helping us apply our approach at scale at Twitter. We thank Tom Everitt for feedback on a draft of the paper.

\bibliography{refs}
\bibliographystyle{unsrtnat}

\clearpage
\appendix

\begin{table*}[h] 
    \centering
    {\renewcommand{\arraystretch}{1.3}
    \begin{tabular}{|l|c|c|c|} \multicolumn{4}{c}{$\prob(V=1 \mid \mathtt{Behavior}=1)$} \\
        \bottomrule
          $\mathtt{Behavior}$ & Naive Bayes & $\Click, \Open \nrightarrow \SLO$ & Full Model \\
          \bottomrule
         $\OptOut$ & -99.74 & -0.932 & -0.072 \\ \hline
         $\Click$ & -1.194 & 0.316 & 0.652 \\ \hline
         $\Open$ & -0.366 & 0.442 & 0.685 \\ \hline
         $\UAM$ & -1.092 & 0.157 & 0.719 \\ \hline
         $\VW$ & -0.475 & 0.254 & 0.772\\ \hline
         $\Linger{6}$ & -0.525 & 0.264 & 0.802 \\ \hline
         $\LC$ & -0.302 & 0.320 & 0.836 \\ \hline
         $\Reply$ & 0.358 & 0.570 & 0.932 \\ \hline
         $\Linger{12}$ & -0.254 & 0.245 & 0.948\\ \hline
         $\Fav$ & 0.579 & 0.672 & 0.949 \\ \hline
         $\RT$ & 0.680 & 0.720 & 0.956\\ \hline
         $\Linger{20}$ & 0.019 & 0.296 & 0.991 \\
         \hline
         $\Quote$ & 1.0 & 1.0 & 1.0 \\
         \bottomrule
    \end{tabular}
    }
    \caption{The same inferences as in Table \ref{tab:bnet_infs}, except without clamping to $[0, 1]$.}
    \label{tab:bnet_infs_unclamped}
\end{table*}

\end{document}